\documentclass[11pt]{llncs} 

\usepackage{hyperref} 
\usepackage{graphicx}
\usepackage{amssymb}

\newcommand{\oneset}{$k$-\textsc{Value 1-Set}}
\newcommand{\maxset}{\textsc{Max} 3-\textsc{Value $r$-Set}}
\newcommand{\minset}{\textsc{Min} 3-\textsc{Value $r$-Set}}
\newcommand{\kmc}{$k$-\textsc{Multicolored Clique}}
\newcommand{\rdm}{\textsc{Perfect Multi-Dimensional Matching}}
\newcommand{\mrdm}{\textsc{Maximum Multi-Dimensional Matching}}
\newcommand{\pack}{\textsc{Set Packing}}
\newcommand{\eds}{\textsc{Independent Edge Dominating Set}}
\newcommand{\tpset}{\textsc{2P 3-Value Set}}
\newcommand{\arc}{\textsc{Arc Kayles}}
\newcommand{\node}{\textsc{Node Kayles}}

\begin{document}

\title{The Computational Complexity of the Game of Set and its Theoretical Applications}
\author{Michael Lampis\inst{1}, Valia Mitsou\inst{2}}
\institute{Research Institute for Mathematical Sciences (RIMS), Kyoto University \\\email{mlampis@kurims.kyoto-u.ac.jp} \and CUNY Graduate Center \\\email{vmitsou@gc.cuny.edu}}

\maketitle

\begin{abstract}

The game of SET is a popular card game in which the objective is to form Sets using cards from a special deck. In this paper we study single- and multi-round variations of this game from the computational complexity point of view and establish interesting connections with other classical computational problems. 

Specifically, we first show that a natural generalization of the problem of finding a single Set, parameterized by the size of the sought Set is W-hard; our reduction applies also to a natural parameterization of \rdm, a result which may be of independent interest. Second, we observe that a version of the game where one seeks to find the largest possible number of disjoint Sets from a given set of cards is a special case of 3-\pack; we establish that this restriction remains NP-complete. Similarly, the version where one seeks to find the smallest number of disjoint Sets that overlap all possible Sets is shown to be NP-complete, through a close connection to the \eds\ problem. Finally, we study a 2-player version of the game, for which we show a close connection to \arc, as well as fixed-parameter tractability when parameterized by the number of rounds played. 

\end{abstract}

\section{Introduction}

In this paper, we analyze the computational complexity of some variations of
the game of SET and its interesting relations with other classical problems,
like \rdm , \pack , and \eds. 

The game of SET is a card game in which players seek to form Sets of cards from
a special deck. Each card from this deck has a picture with 4 attributes
(shape, color, number, shading), and each attribute can take one of 3 values
(for example the shape can be oval, squiggle, or diamond, the color can be
blue, green, or purple, etc). To create a \emph{Set}\footnote{The first letter
of Set is capitalized to avoid a mix-up with the notion of mathematical set},
the player needs to identify 3 cards in which, for each attribute
independently, either all cards agree on the value, or they constitute a
rainbow of all possible values. In a single round of the normal play, 12 cards
are dealt and the players seek (simultaneously) a Set. The first player to find
a Set wins the 3 cards constituting it. Then 3 new cards are dealt in the old
ones' places and the game continues with the next round.
For more information regarding the game and its rules as well as for other variations see the official website of the game \url{http://www.setgame.com/set/index.html}. 

The game of SET has gained remarkable attention and popularity (especially
among mathematicians) as well as many awards. The game has been the subject of
both educational and technical research. A broad set of educational activities
has been suggested, a collection of which can be found 
in \cite{education}.
Furthermore, the game has been studied extensively from a more technical
mathematical point of view, considering questions like ``what is the maximum
number of cards with $n$ attributes and 3 values that can be laid such that no
Sets are formed'' \cite{Davis_thecard}, or ``for fixed $n$, how many
non-isomorphic collections of $n$ cards are there'' \cite{coleman2012game}). In
\cite{Zabrocki_thejoy}, many other similar questions are posed. In addition to
the game's popularity, one motivation for this intense study is that the
problem has a very natural alternative mathematical formulation: if one
describes the cards as four-dimensional vectors over the set $\{0,1,2\}$, then
a Set is exactly a collection of three collinear points, that is, three points
whose vectors add up to $0 (\bmod 3)$.  Nevertheless, the first and - to the
best of our knowledge - only attempt to consider the game's computational
complexity was made by Chaudhuri et al \cite{set} in 2003, who showed that a
generalization of the game is NP-complete. Our focus on this paper is to
continue and refine this work by studying further aspects of the computational
complexity of SET.

In order to study a game from the viewpoint of computational complexity theory,
one needs to define a natural generalization of the game in question (as the
original constant size game always has constant time and space complexity). In
a round of SET, there are 3 parameters to consider: the number of cards $m$,
the number of attributes $n$ and the number of values $k$ (in the original game
$m=12$, $n=4$ and $k=3$). A subset of $k$ cards will be considered to be a Set
if for all attributes, values either all agree or all differ. Of course these
three parameters are not totally independent as the number of cards $m$ is
upper-bounded by $k^n$. In any multi-round version of the game, an extra
parameter $r$ being the number or rounds is added. 

\subsubsection*{Summary of results.}

We first talk about a single-round version of SET. This one-round version
generalizes \rdm\ as was first observed in \cite{set}. It is easy to see that
the problem parameterized by the number of values $k$ is in XP (by the trivial
algorithm that enumerates all size-$k$ sets of cards and checking whether any
of them constitutes a Set). We prove that this parameterized version of the
problem is W-hard. Our W-hardness proof applies to \rdm\ as well, proving that
\rdm\ parameterized by the size of the dimensions $k$ (while the number of
dimensions $n$ is unbounded) is W[1]-hard. This result may be of independent
interest, as this is a natural parameterization of a classic problem that has
not been considered before. The only relevant parameterized result known about
this problem is that \mrdm\ parameterized by the size of the matching and the
number of dimensions is FPT (first established in \cite{parameterizedcomplexity} 
and further improved in
\cite{ChenFLLW11}.

Next, we focus our attention to the case where the number of values is 3. As 
was suggested, there is a polynomial time algorithm to find whether there
exists at least one Set, in other words to play just one round. The complexity
stays the same even if we consider the question of enumerating all Sets. This
generalizes the daily puzzles found either on the
\href{http://www.setgame.com/set/index.html}{official website of SET} or in
\href{http://www.nytimes.com/ref/crosswords/setpuzzle.html}{the New York
Times}.  In these puzzles we are given $m$ cards and need to find the maximum
number of Sets assuming that we don't remove any cards from the table after
finding a Set.

It becomes interesting to ask the same question for a multi-round game, where
cards are gradually removed. This corresponds to the CO-OP version of the game,
where players have to cooperate in order to find the maximum number of
available Sets given that cards of found Sets are removed from the table.
Another interesting variation is the one where we are looking for the minimum
number of Sets that once picked destroy all existing Sets. Both problems can be
seen as special cases of more general packing and covering problems. In the
maximization version, one is looking for a maximum 3-\pack, while in the
minimization version one is looking for a minimum \eds\ in a 3-uniform
hypergraph. We show that both problems remain NP-Hard even on instances that
correspond to the SET game. From the parameterized point of view, if one
considers as the parameter the number of rounds $r$ to be played, a natural
parameterization of the former problem asking whether there are at least $r$
mutually disjoint Sets is Fixed Parameter Tractable, following from the results
of Chen et al.  \cite{ChenFLLW11}. We establish that the natural parameterized
version of the latter problem (find at most $r$ Sets to destroy all Sets) is
also FPT, through a connection with the related \eds\ problem on graphs. 

Finally, we consider a two-player version of the $r$-round game, which can be
seen as a restriction of the game \arc\ in 3-uniform hypergraphs (where
hyperedges should be valid Sets). The complexity of \arc\ is currently unknown
even on graphs and it has been a long-standing open question since the
PSPACE-Completeness of its sibling problem \node\ was established in
\cite{Schaefer78}. We prove that this multi-round 2-player version of SET is at
least as hard as \arc. Nevertheless, we prove that deciding whether the first
player has a winning strategy in $r$ moves in 2-player SET is FPT parameterized
by $r$. This implies the same result for \arc\ on graphs.

The paper is divided as follows: In section \ref{sec:oneset} we present the
W-hardness of the single-round version of SET. In section \ref{sec:multiround}
we analyze the above-mentioned multi-round variations with $k=3$. In section
\ref{sec:2pset} we analyze the natural turn-based 2-player version. Last, in
section \ref{sec:conclusions} we give some conclusions and open problems.

\section{W-hardness of \oneset\ and \rdm} \label{sec:oneset}

In this section, we talk about a single-round generalization of the game of SET. We are dealt $m$ cards, each with $n$ attributes that can take one of $k$ values and we need to find a set of size $k$. This is the main problem considered by Chaudhuri et al. in \cite{set}. Their main insight is that this problem can be seen as a hypergraph problem. Specifically, one may construct a hypergraph on $n\cdot k$ vertices, each representing an attribute-value pair. Now, cards can be represented as hyperedges, by including in each hyperedge the $k$ values that describe the corresponding card's attributes. It is not hard to see that a perfect matching in this $n$-partite hypergraph corresponds to a Set in the original instance. On the other hand, some Sets do not correspond to perfect matchings, because all cards may share the same value for some attributes. Nevertheless, Chaudhuri et al. have established that the two problems have the same complexity and finding a Set is essentially algorithmically equivalent to find a perfect matching in this hypergraph.   

Here we will exploit this connection between the two problems to analyze the complexity of finding a Set with respect to the three relevant parameters $m,n$, and $k$. If $k$ is unbounded, finding a Set was shown to be NP-hard in \cite{set} even for just 3 attributes. If the cards have only 2 attributes, the game is in P. On the other hand, if $n$ is unbounded but the number of values $k$ is considered as a parameter the problem is trivially in XP. Here we will show that the trivial algorithm cannot be improved to an FPT algorithm, by proving that the problem is W[1]-hard. The first step of our reduction is to show that the relevant parameterization of \rdm\ is W[1]-hard, a result that may be of independent interest.

\begin{theorem} \label{thm:Whard}
\rdm\ parameterized by the dimension size is W[1]-hard.
\end{theorem}

\begin{proof}
We present a reduction from \kmc\ (proven to be W[1]-hard in \cite{multicoloredclique}).

Given an instance of \kmc, in other words a $k-$partite graph $G(V,E)$ where each part has size $n$, we construct an instance of \rdm, a multigraph $G'(V',E')$ with $nk(k-1)$ dimensions where each dimension has $k+{k\choose 2}$ different values, such that if $G$ has a clique of size $k$ then $G'$ has a multidimensional perfect matching.

For each ordered pair $(V_i,V_j)$ with $V_i, V_j, i\neq j$ being parts of $V$, we add $n$ dimensions which we group together in a group $i-ij$. Each of the $n$ dimensions in each group $i-ij$ of graph $G'$ corresponds to a different vertex in part $V_i$ of graph $G$. Each dimension will have $k+{k\choose 2}$ different possible values, one value corresponding to each part $V_i$ and one value corresponding to each pair of parts $(V_i, V_j), i<j$. 



\begin{figure}[ht]
\centering
\begin{minipage}[b]{0.45\linewidth}
\centering
\includegraphics[scale=1.2]{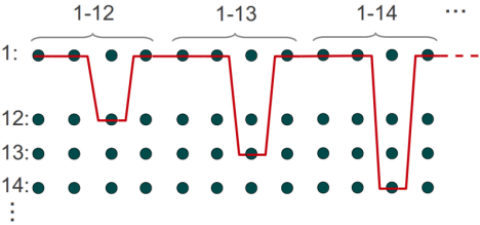}
\caption{The vertex-multiedge of $G'$ that corresponds to vertex $v_{13}$ of part $V_1$ in $G$.}
\label{fig:vertices}
\end{minipage}
\hspace{0.5cm}
\begin{minipage}[b]{0.45\linewidth}
\centering
\includegraphics[scale=1.2]{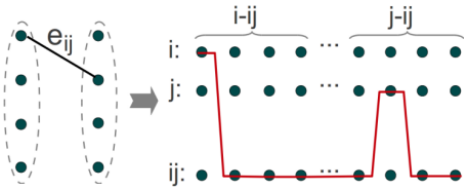}
\caption{The edge-multiedge of $G'$ that corresponds to the edge $e_{ij}$ of $G$.}
\label{fig:edges}
\end{minipage}
\end{figure}

Furthermore, for each vertex $v_{ij}$ in the original graph ($j^{th}$ vertex of part $V_i$) we create a multiedge as follows (see figure \ref{fig:vertices}): it will contain the vertices labeled with $i$ for all dimensions but the $j^{th}$ dimension of each group $i-ki$, where $k\neq i$. For these dimensions we 'll include the vertex labeled with $kj$. We call these \emph{vertex-multiedges}. 

Last, for each edge $e_{ij}\in E$ that connects the $a^{th}$ vertex of part $V_i$ with the $b^{th}$ vertex of part $V_j$ in the original graph, we create a multiedge as follows (see figure \ref{fig:edges}): we add all vertices labeled with $ij$ for all dimensions except for the $a^{th}$ dimension in the group $i-ij$ that take the vertex with label $i$ and the $b^{th}$ dimension in group $j-ij$ that we take the vertex with label $j$. We call these \emph{edge-multiedges}.

Notice that the above construction is polynomial in the size of the input and the parameter of \kmc. Also, the dimension size in the constructed instance of \rdm\ $k+{k \choose 2}$ is quadratic in the parameter $k$ of \kmc.

\begin{figure}[h]
\centering
\includegraphics[scale=1.4]{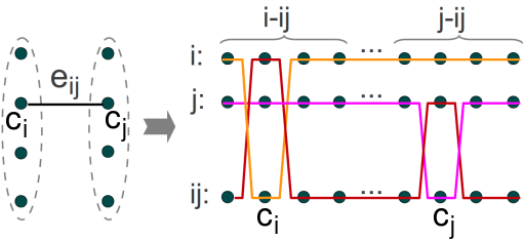}
\caption{Vertices of groups $i-ij$ and $j-ij$ that were not covered by the vertex-multiedges of $G'$ that correspond to vertices $v_{ic_i}$ or $v_{jc_j}$ of $G$ are covered by the edge-multiedge of $G'$ that corresponds to edge $e_{ij} = (v_{ic_i}, v_{jc_j})$ and vice versa.} 
\label{fig:reduction}
\end{figure}

Now we prove that if $G$ has a clique of size $k$ then $G'$ has a perfect multidimensional matching and vice versa. Suppose that $G$ has a clique of size $k$. In other words, there should be a tuple $(v_{1c_1}, v_{2c_2}, \ldots v_{nc_n})$, with $v_{ic_i} \in V_i$, where all vertices in the tuple are connected with each other. We select in the matching the $k$ vertex-multiedges of $G'$ that correspond to the vertices in the clique of $G$ and the $k\choose 2$ edge-multiedges of $G'$ that correspond to edges of $G$ that connect vertices in the clique. This selection is a perfect matching: each vertex-multiedge or edge-multiedge selects all vertices with labels that correspond to the vertex or edge that they represent, except for $k-1$ vertices for each vertex-multiedge and 2 vertices for each edge-multiedge as it is described above. Also, the edge-multiedge of $G'$ that corresponds to edge $e_{ij} = (v_{ic_i}, v_{jc_j})$ of $G$ covers those two vertices that the vertex-multiedges that correspond to $v_{ic_i}$ and $v_{jc_j}$ left uncovered, and vice versa (see figure \ref{fig:reduction}). 

On the other hand, if $G'$ has a perfect matching, then this matching contains exactly one vertex-multiedge and exactly one edge-multiedge of each value (otherwise there would be uncovered vertices or vertices covered twice by the matching). We select all vertices of $G$ that correspond to a vertex-multiedge in the matching. Now, all these vertices that we picked should be pairwise connected in $G$, because the edge-multiedges in the matching should be covering those vertices in $G'$ that the vertex-multiedges didn't cover, which correspond to the vertices in the clique.

For a complete example of the construction see figure \ref{fig:example}. \qed

\begin{figure}[h] 
\centering
\includegraphics[scale=1.4]{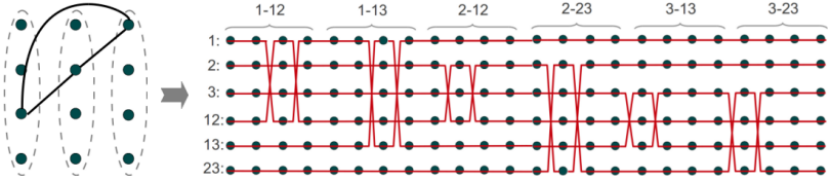}
\caption{A complete example for W-hardness of Section \ref{sec:oneset}.}
\label{fig:example}
\end{figure}

\end{proof}

\begin{corollary} 
The game of Set parameterized by the number of values (or else the size of the Sets) is W[1]-hard.
\end{corollary}

\begin{proof}
The ``if'' part of the above reduction also holds for the game of Set: if $G'$ has a multidimensional perfect matching it also has a Set. For the ``only if'' part, notice that if $G'$ has a Set then this Set is also a multidimensional perfect matching since no vertex-multiedge can pass through a value that belongs to another vertex-multiedge. \qed
\end{proof}

\section{Multi-round variations of SET}\label{sec:multiround}

In this and the next section we talk about multi-round variations of SET where
the number of values (or in other words the size of the Sets) is 3. 
In this case, each card (vertex of the hypergraph) is described by a 
vector in $\mathbb{F}_3^{n}$. Note that, three cards form a Set if and only 
if their corresponding vectors add up to the all-$0$ vector. It is also easy to
observe that every pair of cards can have up to one card that forms a Set with
the other two. This property will prove useful later.

We will once again use a hypergraph formulation, though
different from the one in the previous section. Specifically, we consider the
$3$-uniform hypergraph formed if we construct a vertex for each dealt card and
a hyperedge (that is, a set of size 3) for each Set. It is clear that given a
SET instance, one can in polynomial time construct this hypergraph. 

We will first talk about a maximization variation: given a set of cards we ask the question
whether there exist at least $r$ Sets that we can pick up before leaving no
Sets on the table. We call this problem \maxset. Observe that this problem is a
special case of 3-\pack, which is a known NP-hard problem. We thus need to show
that the problem remains NP-hard when restricted to instances realizable by SET
cards. This is established in Theorem \ref{thm:maxset}.

Then, we turn our attention to a minimization version: given a set of cards, is
it possible by removing at most $r$ Sets ($3r$ cards) to eliminate all potential Sets?
We call this problem \minset. This problem is a special case of \eds\ in 3-uniform hypergraphs. We show its
NP-hardness even when restricted to hypergraphs realizable by SET cards. Then, we prove that 
the natural parameterized version of \eds\ in 3-uniform hypergraphs with parameter $r$ is FPT, 
thus proving that the special case of a parameterization of this version of SET
is also FPT.

\subsection{NP-Hardness the maximization version}\label{sec:maxset}

\begin{theorem} \label{thm:maxset} 
\maxset\ is NP-Hard.  
\end{theorem}

\begin{proof}

We design a reduction from 3-SAT. Given a formula $\phi$ of 3-SAT we first
create an equivalent formula $\phi$' where each clause contains at most 3
literals and each variable appears exactly 3 times (two as positive and one as
negative or two as negative and one as positive). Furthermore, any two clauses
of $\phi$' share at most one variable. A similar construction appears in
\cite{papad}, but it is also presented below for the sake of completeness. 

\begin{lemma} \label{specialSAT}
Any formula $\phi$ of regular \emph{3-SAT} can be transformed into an equivalent formula $\phi$', where each clause has at most 3 variables and each variable appears exactly 3 times in $\phi$' (not all positive or all negative).
\end{lemma}

\begin{proof}
Given a formula $\phi$ of regular 3-SAT, we create an equivalent formula $\phi$' as follows: first, we ensure that each variable appears at least 4 times (if not, we double some of the clauses where this variable appears); then, for each appearance of each variable $v$ we create a new variable $v_i$ for $i=1, \ldots l$, where $l$ is the total number of appearances, and clauses $(\neg v_i \vee v_{i+1})$, $(\neg v_l \vee v_1)$.

Clearly all variables in $\phi$' appear exactly 3 times and not all positive or all negative. Furthermore, $\phi$ is satisfiable iff $\phi$' is satisfiable by an assignment that sets the same truth value to all variables $v_i$ in $\phi$' corresponding to the same variable $v$ in $\phi$. \qed

\end{proof}

Let $m$ be the number of clauses of $\phi$' and $n$ the number of
variables.

The main idea of the reduction is as follows: from formula $\phi$' we create an
instance of \maxset\ which consists of variable gadgets (one corresponding to
each variable) and clause gadgets (one corresponding to each clause). The
variable gadget of a variable $x$ contains five cards: three cards $x_{1}$,
$x_{2}$ and $x_{3}$ for each appearance of $x$ in $\phi$' ($x_{1}$ and $x_{2}$
corresponding to appearances with the same sign and $x_3$ to opposite), and two
more cards: $x_{12}$ which forms a Set with $x_1$ and $x_2$, and $x_{123}$
which forms a Set with $x_{3}$ and $x_{12}$. Picking either Set is equivalent
to making an assignment to $x$ (both Sets contain $x_{12}$, only one Set can be
formed leaving either positive or negative appearances of $x$ unused). The
cards $x_1,x_2,x_3$ will also appear in the clause gadgets and, intuitively, we
will be able to select a Set from a clause gadget if and only if one of its
$x_i$ vertices is free, corresponding to a true literal.

\begin{figure}[ht]
\begin{minipage}[b]{0.45\linewidth}
\centering
\includegraphics[scale=0.8]{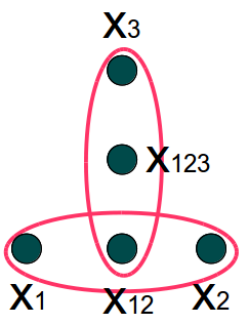}
\caption{The variable gadget}
\label{fig:variables}
\end{minipage}
\hspace{0.5cm}
\begin{minipage}[b]{0.45\linewidth}
\centering
\includegraphics[scale=0.8]{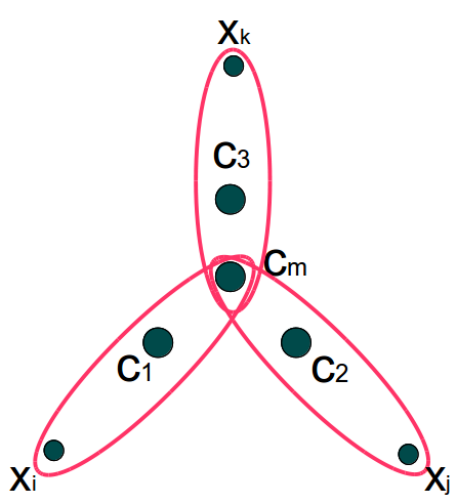}
\caption{The clause gadget}
\label{fig:clauses}
\end{minipage}
\end{figure}

The clause gadget consists of four additional cards: one card per literal in
the clause $c_1$, $c_2$, and $c_3$, and one additional card $c_m$ (for clauses
of size 2 we do not introduce $c_3$).  Furthermore, each card $x_{c_i}$
corresponding to the literal in the $i^{th}$ position of a clause $c$ forms a
Set with cards $c_i$ and $c_m$. In order to be able to pick this Set (and
satisfy $c$) $x_{c_i}$ should not have been picked during the assignment phase. 


Observe that, if one sees the
new instance as a 3-\pack\ instance, it is not hard to establish that the
instance has a solution of size $n+m$ if and only if $\phi$' is satisfiable.
The bonus point is that this instance is realizable with Set cards. In 
what follows we focus our attention to proving this fact. \qed

Each card will be described by a vector in $\mathbb{F}_3^{m+n+1}$. 
The first $n+1$ coordinates constitute
the variable part and the last $m$ the clause part. The variable part is the
same for all cards in each variable gadget representing variable $i$: it
consists of all $0$s, except for the $i^{th}$ coordinate which is set to $1$.
Similarly, vectors of clause gadgets have the same clause part: again all $0$s,
except the ${(n+1+j)}^{th}$ coordinate is set to $1$ for the $j^{th}$ clause.  We
have now fully specified the vectors for the $x_i$'s. Let us explain how the
remaining vectors are filled out.

\begin{itemize}

\item $x_{12}$: clause part is equal to the clause part of $-x_1 - x_2$, so
that $x_1+x_2+x_{12} = 0^m \bmod{3}$;

\item $x_{123}$: clause part is equal to clause part of $- x_3 - x_{12}$; 

\item $c_{m}$: variable part is equal to variable part of $x_{c_1} + x_{c_2} +
x_{c_3}$, if they exist.  If clause has only two literals, we only use
$x_{c_1}+x_{c_2}$ for the first $n$ coordinates while coordinate $n+1$ is set
to 1. The intuition behind introducing the dummy 1 at position $n+1$ for
clauses of size 2 is that it will be convenient if we always know that the
variable part of $c_m$ has three $1$'s.

\item $c_{1}$: variable part is equal to  variable part of $x_{c_1} - x_{c_2} -
x_{c_3}$ ($c_{2}$, $c_{3}$ are formed accordingly). Again, if $x_{c_3}$ does
not exist we use $x_{c_1} - x_{c_2}$ and set coordinate $n+1$ to 2. 

\end{itemize} 

For a detailed presentation of the values of the different types of cards see
table \ref{table}.

\begin{table}
\centering
\begin{tabular}{|c|c|c|}
\hline
\textbf{Card} & \textbf{Variable Part} & \textbf{Clause Part} \\
\hline
$x_i$ & (0, 0, \ldots 0, 1, 0, 0 \ldots 0) & (0, 0, \ldots 0, 1, 0, 0 \ldots 0) \\
$x_{12}$ & (0, 0, \ldots 0, 1, 0, 0 \ldots 0) & (0, 0, \ldots 0, 2, 0, 0 \ldots 0 2, 0 \ldots 0) \\
$x_{123}$ & (0, 0, \ldots 0, 1, 0, 0 \ldots 0) & (0, 1, 0, \ldots 0, 2, 0, 0 \ldots 0 1, 0 \ldots 0) \\
$c_{i}$ & (0, 1, \ldots 0, 2, 0, 0 \ldots 0 2, 0 \ldots 0) &(0, 0, \ldots 0, 1, 0, 0 \ldots 0) \\
$c_m$ & (0, 1, \ldots 0, 1, 0, 0 \ldots 0 1, 0 \ldots 0) & (0, 0, \ldots 0, 1, 0, 0 \ldots 0)\\
\hline
\end{tabular}
\caption{A synopsis of all possible tuples of the different types of card values for proof of Theorem \ref{thm:maxset}.}
\label{table}
\end{table}

Now, we prove that the only Sets which are formed are indeed the Sets that we described
in the introduction of Section \ref{sec:maxset}. To achieve this we need to prove the following 3 Lemmata:

\begin{lemma}\label{lem:same} 
Cards of formed Sets share either the same variable part or the same clause part.
\end{lemma}

\begin{proof} 

First, observe that if two vectors agree in either the clause or the variable
part then the third vector should also agree with them. Therefore, we will only
consider Sets that contain a card of type $c_i$ or $c_m$, because in a Set
containing only cards from the variable gadgets, their vectors should agree
on the variable part.

Suppose that there exists a Set where the 3 cards share neither their variable
part nor their clause part. Since a card of type $c_i$ (or $c_m$) is part of
this Set, then a card of type $c_m$ (or $c_i$ accordingly) should also be part
of it (each of these two cards has three non-zero values in their variable part
and there is no other way to match them with two other cards from variable
gadgets which have only one non-zero value). So this Set should contain a card
of type $c_i$ and a card of type $c_m$. 

Since the two cards we have ($c_i$ and $c_m$) do not agree on their clause
part, the third card of a Set must have exactly two coordinates set to 2 in its
clause part, and all others to 0. Therefore, it must be of type $x_{12}$. The
two 2s of card $x_{12}$ should be aligned with the 1s from $c_i$ and $c'_m$,
when $c$ and $c'$ are different clauses.  But for variable parts to agree,
non-zero values in cards $c_i$ and $c'_m$ should be aligned, which means that
clauses $c$ and $c'$ should contain identical variables.  However that is not
possible from the construction of $\phi$' where different clauses share no more
than one common variable. \qed

\end{proof}

\begin{lemma} \label{lem:variable} 
Only two different types of Sets are formed
by cards that share the same variable part and they intersect.  
\end{lemma}

\begin{proof} 

By construction, there are two different Sets formed within a variable gadget
as shown in figure \ref{fig:variables}. Furthermore, each pair of cards $a$,
$b$ has a unique third card $-(a+b) \bmod 3$ with which they form a Set. Only
possible triplet where cards are pairwise not in participation of existing Sets
are cards $x_1$ (or equivalently $x_2$), $x_3$, and $x_{123}$ which can't form
a Set. \qed
\end{proof}

\begin{lemma} \label{lem:clause}
Sets of cards that share the same clause part shall contain a card of type $c_m$.
\end{lemma}

\begin{proof} Cards of the same clause type are $x_i$, $c_i$ and $c_m$. A card of type $c_i$ can't exist alone with two cards of type $x_i$ because its variable part has three non-zero values and can't match with two cards where each of them has only one non-zero value. Trying to put two cards of type $c_i$ in the same Set won't work either: at least one pair of 2s should be aligned, which means that the last card should also have a 2 in that position. This only leaves a third card of type $c_i$ as a possibility (no other type has a 2 in the variable part). The only way three cards of this type could potential match is if all non-zero values are matched, which would produce three identical cards. \qed
\end{proof}

Observe now that if $\phi$' is satisfiable, then we can select one Set from
each variable gadget (using the corresponding variable's assignment) and one
Set from each clause gadget (since one of the literals is set to True). This
gives $n+m$ Sets. For the converse direction, observe that, from Lemmata
\ref{lem:variable} and \ref{lem:clause} it is not possible to select more than
one Set from each gadget. Thus, one can extract a satisfying assignment for
$\phi$' from a solution of size $n+m$. \qed

\end{proof}

\subsection{Results on the minimization version}

Next, we present yet another multi-round version of SET, \minset. 
We remind the reader that in this problem a single player is trying to remove 
the smallest possible number of Sets so
that no more Sets are left on the table. Each card, as before, has an unbounded number of
attributes and each attribute can take 3 values.

We prove that \minset\ is NP-hard via a simple reduction from \eds\ (proven NP-hard in \cite{garey}).

\begin{theorem} \label{thm:minset}

\minset\ is NP-hard.

\end{theorem}

\begin{proof}

Given an instance of \eds\ (a graph $G(V,E)$ and a number $r$), we create an
instance of \minset\ of $|V|+|E|$ cards with $|V|$ dimensions each, such that
if $G$ has an edge dominating set of size at most $r$ then there exist at most
$r$ Sets which once picked up destroy all other Sets. Again, cards will be
represented by vectors in $\mathbb{F}_3^{|V|}$.

The construction is as follows: For each vertex $i\in V$ we create a card where
all coordinates are 0 except from the value of the $i^{th}$ coordinate which is
equal to 1.  Furthermore, for each edge $(i,j) \in E$ we create a card where
all coordinates are 0 except from the values of coordinates $i$ and $j$ which
are equal to 2. 

Observe that the only Sets formed correspond directly to edges in $G$. Picking
a Set corresponding to edge $(i,j)$ eliminates the cards corresponding to
vertices $i$, $j$ (together with the card corresponding to edge $(i,j)$). This
move causes the elimination of any potential Set containing cards corresponding
to vertices $i$ and $j$. Thus an edge dominating set of size at most $r$ in $G$
corresponds to an equal number of Sets overlapping all other Sets. On the other
hand the smallest number of Sets that overlap all other Sets is equal to the
minimum edge dominating set. \qed

\end{proof}

Since the \minset\ problem is hard, it makes sense to consider its naturally
parameterized version: Given an arbitrary set of cards, do there exist $r$ Sets
that overlap all other formed Sets? We show that a simple FTP algorithm
can decide this question. As a matter of fact, the algorithm works on any 3-uniform 
hypergraph. Recall that the similar
parameterization of the maximization problem is also known
to be FPT, by relevant results on 3-\pack\ \cite{ChenFLLW11}.  

\begin{theorem} \label{thm:edsfpt}
\eds\ in 3-uniform hypergraphs parameterized by the size of the edge dominating 
set is FPT.
\end{theorem}

\begin{proof}

We give an algorithm that follows the same basic ideas as the FPT algorithm for
\eds\ given in \cite{Fernau06}. We will not worry too much about optimizing the
parameter dependence, instead focusing on establishing fixed-parameter
tractability.

Consider the 3-uniform hypergraph formed as follows: we have a vertex for every
given card and a hyperedge of size 3 for each Set of the input instance.
Suppose that there exists a set of $r$ Sets such that removing the cards they
consist of would destroys all Sets. Then, there must exist a hitting set in this
hypergraph of size exactly $3r$ (since the $r$ removed Sets cannot overlap).

We will list all hitting sets of size $3r$ with a simple branching algorithm as
follows: start with an empty hitting set and as long as the size of the
currently selected hitting set has size $<3r$ find a hyperedge that is
currently not covered. For each non-empty subset of the vertices of this
hyperedge (there are 7 choices) add these vertices to the hitting set and
remove all hyperedges they hit. Recursively continue until either all
hyperedges are hit or the hitting set has size more than $3r$. If we have a
hitting set of size exactly $3r$ add it to the list.

For each hitting set $S$ of size exactly $3r$ do the following: check if the
hypergraph induced by $S$ has a perfect matching, that is, a set of $r$
disjoint hyperedges covering all vertices. This can be done in time exponential
in $r$. If the answer is yes, we have found a set of $r$ Sets that overlaps all
other Sets. If the answer is no for all hitting sets then we can reject. \qed

\end{proof}

\begin{corollary} \label{cor:setfpt}
\minset\ parameterized by the number of Sets that will be picked is FPT.
\end{corollary}

Corollary \ref{cor:setfpt} follows directly from Theorem \ref{thm:edsfpt}.

\section{A two player game} \label{sec:2pset}

In this section, we consider a natural two-player turn-based game 
that we call \tpset.  
Suppose that an arbitrary set of cards is on the table and two opposing
players take turns playing. Each player may select three cards that form a Set
and remove them from play. No additional cards are dealt. The game goes on
until a player is unable to find a Set, in which case she loses. 

Unlike the solitaire games \maxset\ and \minset, here players must exercise some strategic
thinking: each is trying not only to maximize the number of Sets she will
collect but also to prevent the opponent from forming a set.

We exploit the ideas developed for the single-player game \minset. Although 
we will not completely settle the complexity of the two-player version, the
reduction given in Theorem \ref{thm:minset} can be used to establish directly
that the two-player version of Set is at least as hard as \arc.

\arc\ is a two-player game played on an undirected graph. Two players take
turns selecting edges from the graph, under the constraint that the edge they
pick cannot share a common endpoint with any previously selected edges.  The
first player unable to move loses. 

Though the complexity of the related version of the problem called \node\ 
was settled in the '70s by Schaefer \cite{Schaefer78}, \arc\ has been
open ever since. It is not hard to see that, since the game in \arc\ ends
essentially when the two players have formed a minimal independent 
edge dominating set, we can say the following:

\begin{corollary}

\tpset\ is at least as hard as \arc.

\end{corollary}

It will likely be hard to find a polynomial-time algorithm for \arc, and
therefore also for \tpset. A slightly more general version of \arc\ is
mentioned to be PSPACE-complete in \cite{Schaefer78}, while the natural
generalization of \arc\ to hypergraphs with unbounded hyperedge size is
PSPACE-hard by the complexity of poset games \cite{Grier13}. 

The 2-player SET problem on graphs is a natural restriction of \arc,
though this version of SET, unlike its hypergraph counterpart turns out 
to be trivial: if the size of the Sets (i.e. the number 
of different values) is 2 then any 2 cards form a 
Set; thus the 2-player problem is equivalent to \arc\ on complete graphs 
and becomes a simple matter of parity of the number of nodes. 

Let us consider a natural parameterization of \tpset. In this problem, the
question is whether a winning outcome for the first player can be achieved
within at most $r$ rounds (with $r$ being the parameter). Parameterized
problems of this form have been considered in the past, beginning with
\cite{AbrahamsonDF95}, where it was established that the $r$-move parameterized
version of \node\ is AW[*]-hard. \tpset\ (and thus \arc\ too), as we show in
Theorem \ref{thm:2p-FPT}, parameterized by the number of rounds turns out to be
FPT.

\begin{theorem} \label{thm:2p-FPT}
\tpset\ parameterized by the number of allowed rounds $r$ is FPT.
\end{theorem}

\begin{proof}

First, observe that hypergraph $G$ where the game is played should have an edge
dominating set of size at most $r$ and thus a hitting set of size at most $3r$.
If there is no hitting set of size at most $3r$, simply reply no because it's
then impossible for the first player to end the game in $r$ moves.  Otherwise
we compute such a hitting set.  This can be done in FPT time 
\cite{wahlstrom2007algorithms}.
Name the vertices of the hitting set $h_1, h_2, \ldots, h_s$, where $s$ is the 
size of the hitting set.

We can now reduce our problem to an ordered version of \node\ on an $r$-partite
graph. In this version the input is an undirected simple graph $G'(V,E)$ where
$V$ is partitioned into $r$ independent sets $V_1,\ldots,V_r$. The two players
alternate turns, and in turn $i$ the current player must select a vertex from
$V_i$ so that it has no edges to previously selected vertices.

We can construct $G'$ from $G$ as follows: for each hyperedge $e$ of $G$
construct $r$ vertices $e_1,\ldots,e_r$ in $G'$, such that $e_i\in V_i$ for all
$i$. If two hyperedges $e,f$ share an endpoint in $G$ connect the vertices
$e_i,f_j$ for all $i\neq j$. It is not hard to see that player 1 has a winning
strategy in the new game if and only if he has a winning strategy of length at
most $r$ in the original game.

We will say that a vertex $e_i$ of $G'$ has color $j$ when the hitting set
vertex $h_j$ is contained in the hyperedge $e$. Notice that all vertices of
$G'$ have some color, and none can have more than three. Also, for any pair of
colors $i,j$ there is at most one vertex in each partite set that has both
colors $i$ and $j$, since by the Set property any two vertices of the original
hypergraph have a unique third vertex with which they form a Set. Finally, note
that for each $i$ such that $ 1\le i \le s$, the vertices with color $i$ form an
$r$-partite complete subgraph in $G'$, since they all come from hyperedges that
contain $h_i$. An example of the construction appears in figure \ref{fig:2ptrans}.

\begin{figure}
\centering
\includegraphics[scale=1]{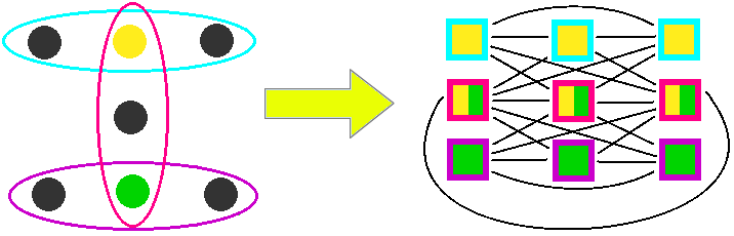}
\caption{An example of the construction of the $r-$partite graph $G'$ from 
hypergraph $G$ for $r=3$.}
\label{fig:2ptrans}
\end{figure}

Partition the set $V_r$ into subsets such that each set contains vertices with
exactly the same colors. The subsets where vertices have two or three colors
have, as we argued, size 1. Consider now the subset of vertices
$S_{r,i}\subseteq V_r$ which have color $i$ only. We first have the following:

\begin{claim} \label{claim:degree}

A vertex $e_j\in V_j$, with $j\neq r$ and $e_j$ not having color $i$ can have
at most $2$ neighbors in $S_{r,i}$.  

\end{claim}

\begin{proof}

Suppose that $e_j$ has three distinct neighbors in $S_{r,i}$. Let the
hyperedges corresponding to these vertices be $\{h_i,u_1,v_1\},\
\{h_i,u_2,v_2\},\ \{h_i,u_3,v_3\}$. Notice that $h_i$ is the only hitting set
vertex in these sets, as $i$ is the only color of vertices in $S_{r,i}$. Now,
$e_j$ contains $h_j$ which is distinct from $h_i$ and all other vertices in
these sets. So, in order for $e_j$ to intersect all three of these sets, two of
them must share a common vertex other than $h_i$. But this contradicts the Set
property that any two elements have a unique third with which they form a Set.
\qed
\end{proof}

From Claim \ref{claim:degree} we now know that if $S_{r,i}$ contains at least
$2r$ vertices, then it will be possible to play it if and only if no vertex with
color $i$ is played in the first $r-1$ moves. Perform the following
transformation: delete all vertices of $S_{r,i}$ and replace them with a single
vertex that is connected to all vertices in other partite sets that have color
$i$.

The above reduction rule is safe. To see this, consider any play of the first
$r$ moves. If a vertex of color $i$ is used, no vertex from $S_{r,i}$ can be
used in the last move in both graphs. If color $i$ is not used, some vertex of
$S_{r,i}$ can be used in both graphs, and it is immaterial which will be played
since this is the last move.

Because of the above we can now assume that $|S_{r,i}|\le 2r$. Thus,
$|V_r|=O(r^2) $, because we have $s=O(r)$ sets $S_{r,i}$, as well as the single
vertices which may have a pair of colors.

We will now move on to the preceding partite sets using a similar argument. We
need the following definition: if two vertices $e_i,f_i$ of same color $c$ in the 
same partite set $V_i$ have exactly the same neighbors in all sets $V_j$ for $j>i$, then they
are called equivalent. We call such vertices equivalent because, if both
are available to be played at round $i$ they can be selected interchangeably
without affecting the rest of the game. Observe, that equivalent vertices have 
the same neighbors in $V_j$ for all $j\ge i+1$. Also observe that each equivalence
class cannot have more than $2$ neighbors. Namely, if two vertices of same color 
$c$ have at least one common neighbor then from claim \ref{claim:degree} this 
common neighbor cannot have more than these two vertices as common neighbors from 
color class $c$. On the other hand, if two or more vertices of color $c$ both have 
no neighbors, then we can all merge them into a single vertex.

We will use this fact to show that we can reduce the graph so that in the end
$|V_i|\le |V_{i+1}|^{O(r)}$. 
Initially it may appear that the argument would lead to the conclusion
that $|V_i|\le 2^{|V_{i+1}|}$, since we have a different equivalence class of
each possible set of neighbors that a vertex of $V_i$ can have in $V_{i+1}$.
However, observe that each vertex of $V_i$ can have at most $2s$ neighbors with
which it does not share a color in $V_{i+1}$, since from Claim
\ref{claim:degree} it can have at most 2 neighbors in each group that
correspond to a different color. Thus, the possible neighborhoods are at most
${|V_{i+1}| \choose 2s} = |V_{i+1}|^{O(r)}$.

From the above it follows that the order of $G'$ after applying the above
preprocessing exhaustively is $2^{2^{O(r)}}$, which gives a kernel. \qed

\end{proof}

The proof only uses the property of SET that every pair of cards has a unique 
third that forms a Set with them. Thus the game is FPT even when played on the 
more general class
of 3-uniform hypergraphs having this property. Also, Corollary \ref{thm:arcfpt} 
follows directly from Theorems \ref{thm:minset} and \ref{thm:2p-FPT}:

\begin{corollary} \label{thm:arcfpt}
The natural parameterization of \arc\ by the number of rounds played is FPT.
\end{corollary}
 
The proof of Theorem \ref{thm:2p-FPT} gives a doubly exponential parameter dependence.
Below we present a simpler algorithm which also implies a better complexity.

\begin{proof} 
(\emph{Sketch.} ) 
First, observe that graph $G$ where the game is played should have a vertex
cover of size at most $s=2r$. If not, reply no. The remaining vertices
forming an independent set can be divided into $2^{s}$ equivalence classes depending on
their neighbors in the vertex cover. 

If an equivalence class is large enough, playing any edge in an equivalence
class can be replaced by playing any other from the class without affecting the
rest of the game.  Namely, if an equivalence class is joined to $t$ vertices in
the vertex cover, there can be at most $t\leq s$ vertices played from this
class and it is unimportant which ones are played.  Thus, if a class has more
than $t$ vertices we can simply leave it with $t$ vertices and delete the rest.
 
Because of the above we have $2^s$ groups of vertices each containing at most
$s$ vertices and a vertex cover of size $s$. This means that the graph contains
at most $2^{2s}$ edges. Since in each turn a player selects an edge the number
of possible plays is at most $(2^{2s})^s = 2^{O(r^2)}$. Simply enumerating them
all gives an FPT algorithm. \qed

\end{proof}

\section{Conclusions and Open Problems} \label{sec:conclusions}

In this paper we studied the computational complexity of the game of SET and
presented some interesting connections with other well-studied problems, such
as \rdm, \eds\ and \pack. 

The one-round case of SET is now fairly well-understood. However there are
quite a few interesting open problems one might consider in the multi-round
case, especially the two-player version \tpset.  It remains unknown whether this game is
PSPACE-Complete. However, proving the hardness of \arc\ on graphs would settle the complexity
of this problem as well (which is an interesting open question on its own
accord). Staying on \arc, it might be interesting to show whether the game played
on general 3-uniform hypergraphs is FPT. We remind the reader that our proof that 
\tpset\ is FPT is based on the property of SET that each pair of cards can have at 
most one third with which they all form a Set. That property is vital 
for the proof since it establishes that the line graph has essentially bounded degree. 
This is not true for a general 3-uniform hypergraph though.

\bibliography{bibl}

\newpage
\appendix

\section*{Appendix}
\subsection*{Definitions}

Here, we give the definitions of the problems that we use for the convenience of the reader of this manuscript and for the sake of completeness.


\noindent\rdm: 

\begin{itemize}
\item[] Input: A multigraph $G(V,E)$, with $V = V_1\cup V_2\cup \ldots\cup V_r$ and $|V_i| = k$ for all $i= 1, \ldots r$, and $E \subset V_1\times V_2 \times \ldots V_r$.
\item[] Question: Does there exist a perfect matching in $G$? I.e, does there exist a set of $k$ disjoint multiedges $\{e_1, e_2, \ldots e_k\}$ such that $\bigcup_{i} e_i = V$?
\item[] We call each $V_i$ a \emph{dimension} of $G$. There are $r$ dimensions in $G$ and each dimension has $k$ different possible \emph{values}.\item[] In the parameterized version that we consider, the parameter is $k$.
\end{itemize}

\noindent\pack:

\begin{itemize}
\item[] Input: A 3-uniform hypergraph $G(V,E)$ and a natural number $k$.
\item[] Question: Does $G$ have a set packing of size $k$? In other words, does there exist a set of disjoint hyperedges $E'\subset E$ with $|E'|\geq k$?
\end{itemize}

\noindent\eds:

\begin{itemize}
\item[] Input: A (hyper)graph $G(V,E)$ and a natural number $k$.
\item[] Question: Does $G$ have an independent edge dominating set of size $k$? In other words, does there exist a disjoint set of (hyper)edges $E'\subset E$ with $|E'|\leq k$ such that every (hyper)edge $e\in E$ shares at least one end-point with one or more (hyper)edges in $E'$?
\end{itemize}




\noindent\kmc:

\begin{itemize}
\item[] Input: A $k-$partite graph $G(V,E)$, with $|V_i| = n$ for all $i=1 \ldots k$, $V_1, V_2, \ldots V_k$ pairwise disjoint, and $V = V_1\cup V_2\cup \ldots \cup V_k$.
\item[] Question: Does $G$ have a clique of $k$ vertices? In other words, does there exist a tuple $(v_1, v_2, \ldots, v_k) \in V_1\times V_2 \times \ldots V_k$, such that $(v_i, v_j) \in E$ for all $i \neq j$?
\item[] Parameter: $k$
\end{itemize}

\noindent 3-SAT
\begin{itemize}
\item[] Input: A logic formula $\phi$ written in CNF that contains $n$ variables and $m$ clauses, where each clause contains at most 3 literals.
\item[] Question: Does, does there exist an assignment of truth values to the variables such that all the clauses of $\phi$ are satisfied?
\end{itemize}

\noindent\arc:
\begin{itemize}
\item[] Input: A (hyper)graph $G(V,E)$. 
\item[] Rules: Two players take turns in picking (hyper)edges from $E$ such that picked (hyper)edges don't share endpoints. Player A starts. First player left without an available (hyper)edge to pick loses.
\item[] Question: Is there a winning strategy for player A? 
\end{itemize}

\noindent\node:
\begin{itemize}
\item[] Input: A graph $G(V,E)$. 
\item[] Rules: Two players take turns in picking vertices from $V$ such that picked vertices form an independent set. Player A starts. First player left without an available vertex to pick loses.
\item[] Question: Is there a winning strategy for player A? 
\end{itemize}

\end{document}